\documentclass[acmsmall]{acmart}

\usepackage[notheorems]{nikolako-math}
\usepackage{varwidth}
\usepackage{algorithm}
\usepackage{algpseudocode}
\usepackage{tabularx}
\usepackage{multicol}
\usepackage{booktabs} 
\algblock{ParFor}{EndParFor}
\algnewcommand\algorithmicparfor{\textbf{parfor}}
\algnewcommand\algorithmicpardo{\textbf{do}}
\algnewcommand\algorithmicendparfor{\textbf{end\ parfor}}
\algrenewtext{ParFor}[1]{\algorithmicparfor\ #1\ \algorithmicpardo}
\algrenewtext{EndParFor}{\algorithmicendparfor}

\newif\ifblackandwhite

\usepackage{etoolbox}
\usepackage{pdflscape}
\usepackage{colortbl}%
\newcommand{\myrowcolour}{\rowcolor{white!89.803921568627459!black}}
\usepackage[caption=false]{subfig}

\usepackage[flushleft]{threeparttable}

\definecolor{Maroon}{cmyk}{0,1,1,0.5}

\ifblackandwhite

\newcommand{\highest}[1]{\textbf{#1}}
\else

\newcommand{\highest}[1]{\textcolor{Maroon}{\textbf{#1}}}%
\fi

\usepackage{pgfplots}
\pgfplotsset{compat=newest}
\usepgfplotslibrary{groupplots}
\usepgfplotslibrary{polar}
\usepgfplotslibrary{statistics}
\usetikzlibrary{positioning,chains,fit,shapes,calc}

\definecolor{mycolor1}{rgb}{0.9719668012419661,0.4636196661831345,0.4272112904308525}
\definecolor{mycolor2}{rgb}{0.8273618134500182,0.5740127684739248,0.0}
\definecolor{mycolor3}{rgb}{0.576235910262533,0.6647921100239956,0.0}
\definecolor{mycolor4}{rgb}{0.0,0.7281510724816163,0.22080268587160473}
\definecolor{mycolor5}{rgb}{0.0,0.7553354789496858,0.6252760328031264}
\definecolor{mycolor6}{rgb}{0.0,0.7248676308880103,0.8911735834721597}
\definecolor{mycolor7}{rgb}{0.3804515572007965,0.6115286754014934,1.0}
\definecolor{mycolor8}{rgb}{0.857783940581795,0.44743713611909713,0.9846138428363592}
\definecolor{mycolor9}{rgb}{1.0,0.38135337340068687,0.7651117134811273}

\xdefinecolor{GopherMaroon}{HTML}{7A0019}
\xdefinecolor{GopherGold}{HTML}{FFCC33}
\xdefinecolor{GopherLightGold}{HTML}{FFDE7A}
\xdefinecolor{GopherDarkMaroon}{HTML}{5B0013}

\definecolor{Maroon}{RGB}{238,83,34}
\definecolor{Maroon}{RGB}{25,89,121}
\definecolor{bgorange}{RGB}{238,83,0}
\definecolor{bgorangelight}{RGB}{255,173,115}
\definecolor{bggreen}{RGB}{237,255,178}
\definecolor{bgblue}{RGB}{128,255,204}
\definecolor{Maroon}{RGB}{122,0,25}
\definecolor{Gold}{RGB}{255,204,51} 

\usetikzlibrary{shapes}
\definecolor{pinegreen}{cmyk}{0.92,0,0.59,0.25}
\definecolor{royalblue}{cmyk}{1,0.50,0,0}
\definecolor{lavander}{cmyk}{0,0.48,0,0}
\definecolor{violet}{cmyk}{0.79,0.88,0,0}
\tikzstyle{citems}=[circle, draw, thin,fill=Gold, scale=0.8]
\tikzstyle{cusers}=[rectangle, draw, thin,fill=Maroon, scale=0.8]
\tikzstyle{cusers2}=[rectangle, draw, thin,fill=white, scale=0.8]
\tikzstyle{cred}=[circle, draw, thin,fill=Maroon, scale=0.8]
\tikzstyle{cgreen}=[rectangle, draw, thin,fill=lavander, scale=0.8]
\tikzstyle{rpath}=[ultra thick, Maroon, opacity=0.8]
\tikzstyle{gpath}=[ultra thick, royalblue, opacity=0.5]

\makeatletter
\global\let\tikz@ensure@dollar@catcode=\relax
\makeatother

\newcommand{\movielens}{{\it movielens}}
\newcommand{\yahoo}{{\it yahoo}}
\newcommand{\pin}{{\it pinterest}}
\newcommand{\books}{{\it books}}
\newcommand{\movies}{{\it movies\&tv}}
\newcommand{\electronics}{{\it electronics}}

\newcommand{\puresvd}{\textsc{PureSVD}}
\newcommand{\eigenrec}{\textsc{EigenRec}}
\newcommand{\rpb}{\textsc{RP3}b}
\newcommand{\slim}{\textsc{SLIM}}
\newcommand{\vae}{\textsc{Mult-VAE}}
\newcommand{\dae}{\textsc{Mult-DAE}}
\newcommand{\apr}{\textsc{APR}}

\newcommand{\nais}{\textsc{NAIS}}
\newcommand{\recwalk}{\textsc{RecWalk}}
\newcommand{\pr}{\textsc{RecWalk}\textsuperscript{PR}}
\newcommand{\kstep}{\textsc{RecWalk}\textsuperscript{K-step}}

\begin{document}

\setcopyright{acmlicensed}
\acmJournal{TKDD}
\acmYear{2020} \acmVolume{14} \acmNumber{6} \acmArticle{64} \acmMonth{10} \acmPrice{15.00}\acmDOI{10.1145/3406241}

\title{Boosting Item-based Collaborative Filtering via Nearly Uncoupled Random Walks}

\author{Athanasios N. Nikolakopoulos}
\email{anikolak@umn.edu}
\orcid{1234-5678-9012}
\affiliation{%
  \institution{University of Minnesota}
  \streetaddress{499 Walter Library, 117 Pleasant Street SE}
  \city{Minneapolis}
  \state{Minnesota}
  \postcode{55455}
}

\author{George Karypis}
\affiliation{%
  \institution{University of Minnesota}
	\streetaddress{499 Walter Library, 117 Pleasant Street SE}
	\city{Minneapolis}
	\state{Minnesota}
	\postcode{55455}
	}
\email{karypis@umn.edu}

\renewcommand{\shortauthors}{Nikolakopoulos and Karypis}

\begin{abstract}
  \textit{Item-based models} are among the most popular collaborative filtering approaches for building recommender systems. \textit{Random walks} can provide a powerful tool for harvesting the rich network of interactions captured within 
  these models.   They can exploit indirect relations between the items, mitigate the effects of sparsity, ensure wider itemspace coverage, as well as increase the diversity of recommendation lists.  Their potential 
  however, can be hindered by the tendency of the walks to rapidly concentrate towards the central nodes of the graph, thereby significantly restricting the range of $K$-step distributions that can be exploited for personalized recommendations.  In this work we introduce \textbf{RecWalk}; a novel random walk-based method that leverages the spectral properties of \textit{nearly uncoupled Markov chains} to provably lift this limitation and prolong the influence of users' past preferences on the successive steps of the walk---thereby allowing the walker to explore the underlying network more fruitfully. 
  A comprehensive set of experiments on real-world datasets verify the theoretically predicted properties of the proposed approach and indicate that they are directly linked to significant improvements in top-$n$ recommendation accuracy. They also highlight RecWalk's potential in providing a framework for boosting the performance of item-based models.  
  RecWalk achieves state-of-the-art top-$n$ recommendation quality outperforming several competing approaches, including recently proposed methods that rely on deep neural networks.
\end{abstract}



\maketitle

\section{Introduction}
Recommender Systems are information filtering tools that aim to help users navigate through an ocean of available options and support their decision making process regarding what items to buy, what movies to watch, what articles to read etc. Over the past decade they have become an indispensable component of most e-commerce applications as well as content delivery platforms. Top-$n$ recommendation algorithms are a crucial component of most recommender systems. They provide ranked lists of items tailored to the particular tastes of the users, as depicted by their past interactions within the system. 

\textit{Item-based methods} are among the most popular approaches for top-$n$ recommendation~\cite{sarwar2001item,ning2011slim,ning2015comprehensive}. Such methods work by building a model that captures the relations between the items, which is then used to recommend new items that are ``close'' to the ones each user has consumed in the past. Item-based models have been shown to achieve high top-$n$ recommendation accuracy~\cite{sarwar2001item,ning2011slim} while being scalable and easy to interpret~\cite{ning2015comprehensive}. The fact, however, that they typically consider only direct inter-to-item relations can impose fundamental limitations to their quality and make them brittle to the presence of sparsity---leading to poor itemspace coverage and significant decay in performance~\cite{adomavicius2005toward}. \textit{Random-walk-based methods} are particularly well-suited for alleviating such problems. Having the innate ability to relate items that are not directly connected by propagating information along the edges of the underlying item-to-item graph, random walk methods are more robust to the effects of sparsity and they can afford better coverage of the itemspace.  However, their effectiveness in terms of top-$n$ recommendation can be limited by the tendency of the walks to concentrate probability mass towards the central nodes of the graph---thus disproportionately boosting the recommendation scores of popular items in the system. This means that in order to produce high quality recommendations,  traditional random-walk-based techniques are restricted to exploit just the first few steps of the walk that are still influenced by the personalized starting distribution. This is in accordance to the mathematical properties of random walks and it has also been empirically verified that when applied to real-world networks, short-length random walks typically work best~\cite{Cooper:2014:RWR:2567948.2579244,christoffel2015blockbusters,pixie}.  

In this work\footnote{A preliminary version of this work has appeared in~\cite{RecWalk}.} we introduce \textbf{\recwalk}; a novel framework for top-$n$ recommendations that aims to combine the potential of item-based models to discern meaningful relations between the items, with the inherent ability of random walks to diffuse these relations across the itemspace and exploit the rich network of interactions they shape. 
\recwalk\ produces recommendations based on a \textit{random walk with node-dependent restarts}  designed to prolong the influence of the personalized initialization on the successive $K$-step landing probabilities of the walk---thereby eliminating  the need of ending the walks early. Intuitively, this gives the walker ``more time'' to harvest the information captured within the item model before succumbing to the ``pull'' of central nodes. The proposed random walk construction leverages the spectral properties of \textit{nearly uncoupled Markov chains}~\cite{courtois1977decomposability} in order to enforce a \textit{time-scale dissociation} of the stochastic dynamics of the walk towards equilibrium---thus increasing the number of successive landing probabilities that carry personalized  information useful for top-$n$ recommendation. 
The properties of our model are backed by rigorous theoretical analysis of the mixing characteristics of the walk which we empirically verify that are indeed intertwined with top-$n$ recommendation accuracy. A comprehensive set of experiments on real-world datasets showcase the potential of the proposed methodology in providing a framework for boosting the performance of item models. \recwalk\ achieves high recommendation quality outperforming state-of-the-art competing approaches, including recently proposed methods relying on deep neural networks. 

Open source implementation of the method is available at: 
\begin{center}
	\url{https://github.com/nikolakopoulos/RecWalk}
\end{center}

\section{Notation and Definitions}
\label{Sec:Notation}

\subsection{Notation}
Vectors are denoted by bold lower-case letters and they are assumed to be column vectors (e.g., $\mathbf{v}$). Matrices are represented by bold upper-case letters (e.g., $\mathbf{Q} $). The  $j$-th column and the  $i$-th row of matrix $\mathbf{Q}$ are denoted $\mathbf{q}_j$  and $\mathbf{q}^\transpose_{i}$, respectively. The $ij$-th element of matrix $\mathbf{Q}$ is denoted as $Q_{ij}$ or $[\mQ]_{ij}$ when such choice leads to clearer presentation. We use $\operatorname{diag}(\mathbf{Q})$ to refer to the  matrix that has the same diagonal with matrix $\mathbf{Q}$ and zeros elsewhere, and  $\operatorname{Diag}(\mathbf{v})$ to denote the matrix having vector $\mathbf{v}$ on its diagonal, and zeros elsewhere. We use a boldface  $\ones$ to denote a vector all the elements of which are equal to 1 (when we need to specify the dimension of such vector, we do so with a subscript, e.g., $\ones_n$).  Furthermore, we use $\lVert \cdot \rVert$ to denote a norm that---unless stated otherwise---is assumed to be the Euclidean. Sets are denoted with  calligraphic upper-case letters (e.g., $\mathcal{U,V}$).  Finally, symbol $\triangleq$ is used in definition statements.

\subsection{Definitions}
  Let $\set{U} = \{1,\dots,U\}$ be a set of \textit{users} and  $\set{I} = \{1,\dots,I\}$ a set of \textit{items}. Let $\mR\in \mathfrak{R}^{U \times I}$ be the \textit{user-item interaction matrix}; i.e., the matrix whose $ui$-th  element is 1 if user $u$ has interacted with item $i$, and 0 otherwise. Each user $u \in \set{U}$ is modeled by a vector $\vr_u^\transpose \in \mathfrak{R}^{I}$ which coincides with the  corresponding row of the user-item interaction matrix $\mR$; similarly, each item $i \in \set{I}$ will be modeled by a vector $\vr_i \in \mathfrak{R}^{U}$ which coincides with the corresponding column of matrix \mR.  The rows and columns of \mR\ are assumed to be non-zero; i.e., every user has interacted with at least one item, and for every item there exists at least one user who has interacted with it. Finally, we use the term \textit{item model} to refer to a matrix $\mW \in \mathfrak{R}^{I \times I}$ the $ij$-th element of which gives a measure of \textit{proximity} or \textit{similarity} between items $i$ and $j$. 
\label{sec:RecWalk}

\section{Random Walks and Item Models} 
\label{Sec:Motivation}
The fundamental premise of the present work is that combining random walks and item models allows for more effective utilization of the information captured in the item model; considering direct as well as transitive relations between the items, and also alleviating sparsity related problems. However, directly applying random walks on item models can lead to a number of problems that arise from their inherent mathematical properties and the way these properties relate to the underlying top-$n$ recommendation task. 

Imagine of a random walker \textit{jumping} from node to node on an item-to-item graph with  transition probabilities proportional to the inter-item proximity scores depicted by an item model. If the starting distribution of this walker reflects the items consumed by a particular user $u$ in the past, the probability the walker  \textit{lands} on different nodes after $K$ steps provide an intuitive measure of proximity that can be used to rank the nodes and recommend items to user $u$ accordingly. 

Specifically, if \mW\  denotes the item model and 
\begin{equation}
\mS\mathdef\Diag(\mW\ones)^{-1}\mW
\end{equation} the transition probability matrix of the walk, personalized recommendations for each user $u$ can be produced e.g., by utilizing the $K$-\textit{step landing probability distributions} of a walk rooted on the items consumed by $u$:   
\begin{equation}
\vpi_u^\transpose \mathdef \boldsymbol{\phi}_u^\transpose \mS^K, \qquad \boldsymbol{\phi}_u^\transpose \mathdef \tfrac{\vr_u^\transpose}{\lVert\vr_u^\transpose\rVert_1}
\label{model:srw}
\end{equation}
or by computing the stationary distribution of a \textit{random walk with restarts} on \mS, using $\boldsymbol{\phi}_u^\transpose$ as the restarting distribution. The latter approach is the well-known \textit{personalized PageRank} model~\cite{pagerank} with \textit{teleportation vector} $\boldsymbol{\phi}_u^\transpose$ and \textit{damping factor} $p$, and its stationary distribution can be expressed~\cite{langville2011google} as 
\begin{equation}
\vpi_u^\transpose \mathdef \boldsymbol{\phi}_u^\transpose \sum_{k=0}^{\infty}(1-p) p^k \mS^k.
\label{model:simplepr}
\end{equation}   
Clearly, both schemes harvest the information captured in the $K$-step landing probabilities 
\begin{displaymath}
\{\boldsymbol{\phi}_u^\transpose \mS^k\}_{k=0,1,\dots}.
\end{displaymath} In the former case, the recommendations are produced by using a fixed $K$; in the latter case they are computed as a weighted sum of all landing probabilities, with the $k$-step's contribution weighted by $(1-p)p^k$. But, how do these landing probabilities change as the number of steps $K$ increases? For how long will they still be significantly influenced by user's prior history as depicted in $\boldsymbol{\phi}_u^\transpose$?

When \mS\ is irreducible and aperiodic---which is typically the case in practice---the landing probabilities will converge to a \textit{unique} limiting distribution irrespectively of the initialization of the walk~\cite{grimmett2001probability}. This means that for large enough $K$, the $K$-step landing probabilities will no longer be ``personalized'' in the sense that they will become independent of the user-specific starting vector $\boldsymbol{\phi}_u^\transpose$.  Furthermore, long before reaching equilibrium, the usefulness of these vectors in terms of recommendation will start to decay as more and more probability mass gets concentrated to the central nodes of the graph---thereby restricting the number of landing probability distributions that are helpful for personalized recommendation. 
This imposes a fundamental limitation to the ability of the walk to properly exploit the information encoded in the item model.

Motivated by this, here we propose \textbf{\recwalk}; a novel random-walk model designed to give control over the stochastic dynamics of the walk towards equilibrium; provably, and irrespectively of the dataset or the specific item model onto which it is applied. In \recwalk\  the item model is incorporated as a direct item-to-item transition component of a walk on the user-item bipartite network. This component is followed by the random walker with a fixed probability determined by a model parameter that controls the \textit{spectral characteristics} of the underlying walk. This allows for effective exploration of the item model while the influence of the personalized initialization on the successive landing probability distributions remains strong. Incorporating the item model in a walk on the user-item graph (instead of the item graph alone) is crucial in providing control over the mixing properties; and as we will see in the experimental section of this work such mixing properties are intimately linked to improved top-$n$ recommendation accuracy. 

\section{Proposed Method}
\label{Sec:Proposed_Method}
\subsection{The \recwalk\ Stochastic Process} 
We define $\graph{G} \triangleq (\{\set{U},\set{I}\}, \set{E})$ to be the \textit{user-item bipartite network}; i.e., the network with  adjacency matrix $\mA_\set{G} \in \mathfrak{R}^{(U+I)\times(U+I)}$ given by
\begin{equation}
\mA_\graph{G} \mathdef \bmat{ \mzero 		   & \mR \\ 
	\mR^\transpose & \mzero  }.
\end{equation}
Consider a random walker jumping from node to node on $\graph{G}$. Suppose the walker currently occupies a node $c \in \set{U}\cup\set{I}$. In order to determine the next step transition the walker tosses a biased coin that yields heads with probability $\alpha$ and tails with probability $(1-\alpha)$: 
\begin{enumerate}
	\item If the coin-toss yields \textit{heads}, then: 
	\begin{enumerate}
		\item if $c \in \set{U}$, the walker jumps to one of the items rated by the user corresponding to node $c$ uniformly at random;
		\item if $c \in \set{I}$, the walker jumps to one of the users that have rated the current item uniformly at random; 
	\end{enumerate} 
	\item If the coin-toss yields \textit{tails}, then:
	\begin{enumerate}
		\item if $c\in\set{U}$, the walker stays put;
		\item if $c \in \set{I}$, the walker jumps to a related item abiding by an \textit{item-to-item transition probability matrix} (to be explicitly defined in the following section). 
	\end{enumerate} 
\end{enumerate}   
The stochastic process that describes this random walk is defined to be a homogeneous discrete time Markov chain with state space $\set{U}\cup\set{I}$; i.e., the transition probabilities from any given node $c$ to the other nodes, are fixed and independent of the nodes visited by the random walker before reaching $c$. 

\begin{figure}
	\centering
	\includegraphics{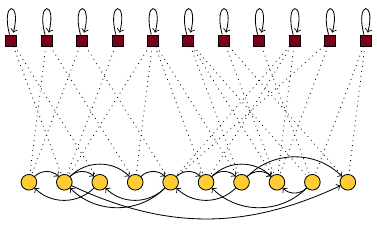}
	\caption{Illustration of the \recwalk\  Model.  \textcolor{GopherMaroon}{\textbf{Maroon}} nodes correspond to users; \textcolor{GopherGold}{\textbf{Gold}} nodes correspond to items. The dashed lines depict transitions between users and items, captured by matrix \mH. The solid lines correspond to node-to-node transitions captured by matrix \mM.}
	\label{fig:Illustration}
\end{figure}

\subsection{The Transition Probability Matrix}
The transition probability matrix \mP\ that governs the behavior of our random walker can be usefully expressed as a weighted sum of two stochastic matrices \mH\ and \mM\ as
\begin{equation}
\mP \mathdef \alpha\mH + (1-\alpha)\mM \label{transitionProbabilityMatrixP}
\end{equation} 
where $ 0< \alpha < 1 $, is a parameter that controls the involvement of these two components in the final model.
Matrix \mH\ can be thought of as the transition probability matrix of a simple random walk on the user-item bipartite network.  
Since every row and column of matrix \mR\ are non zero, matrix \mH\ is well-defined and it can be expressed as   
\begin{equation}
\mH  \mathdef \Diag(\mA_\graph{G}\ones)^{-1}\mA_\graph{G}.
\end{equation} 
Matrix \mM, is defined as 
\begin{equation}
\mM \mathdef \bmat{ \mI & \mzero \\ 
	\mzero & \mM_\set{I} }
\end{equation}
where $\mI \in \mathfrak{R}^{U \times U}$ the identity matrix and  $\mM_\set{I} \in \mathfrak{R}^{I \times I}$ a transition probability matrix designed to capture relations between the items. In particular, given an item model with non-negative weights \mW\  we define this matrix using the following stochasticity adjustment strategy: 
\begin{equation}
\mM_\set{I} \mathdef \frac{1}{\lVert \mW \rVert_\infty} \mW + \Diag(\ones-\frac{1}{\lVert \mW \rVert_\infty}\mW\ones).
\label{def:M_I}
\end{equation}   
The first term divides all the elements by the maximum row-sum of \mW\ and the second enforces stochasticity by adding residuals to the diagonal, appropriately. The motivation behind this definition is to retain the information captured by the relative differences of the item-to-item relations in \mW, ensuring that 
\begin{displaymath}
[\mW]_{ij} \geq [\mW]_{i'j'} \Rightarrow [\mM_\set{I}]_{ij} \geq  [\mM_\set{I}]_{i'j'}, \qquad \text{for all } i \neq j,i' \neq j'.
\end{displaymath}
This prevents items that are loosely related to the rest of the itemspace to disproportionately influence the inter-item transitions and introduce noise to the model.   

An illustration of the \recwalk\ model is given in Fig~\ref{fig:Illustration}.

\subsection{Choice of the core Item-model} 
The construction of matrix \mW\ itself can be approached in several ways depending on the available information, the characteristics of the underlying recommendation problem, the properties of matrix \mR, etc. 
The fact that random walk methods can achieve naturally itemspace coverage  allows us to define this component in a way that promotes sparsity in the relations between the items, having also the advantage to be easy to compute. 

In particular, we propose the use of a locally restricted variant of the well-known SLIM method~\cite{ning2011slim} that is forced to consider only \textit{fixed-size neighborhoods} when learning relations between the items. Concretely, for any given item $i$ we find the set of its $C$ closest neighbors (in terms of cosine similarity between their vector representations) 
and we form a matrix $\mN_i \in \mathfrak{R}^{U\times C}$, by selecting the corresponding columns of the initial matrix $\mathbf{R}$. We then solve for each item the optimization problem
\begin{equation}
\MINone{\vx \in \mathfrak{R}^{C}}{\frac{1}{2}\lVert\vr_i - \mN_i\vx \rVert^2_2 + \gamma_1 \lVert \vx\rVert_1 + \frac{1}{2}\gamma_2 \lVert \vx \rVert^2_2}
{\vx\geq \mzero} \label{myopt}
\end{equation}
and we fill the corresponding elements in the $i$-th column of matrix \mW. 

The estimation of \mW\ can be performed column-wise, in parallel, and it can be implemented efficiently by exploiting the sparsity of $\mN_i$, e.g., using coordinate descent~\cite{slim}. Note that the size of the per-column optimization problems is controlled by the choice of the numbers of neighbors to be considered (i.e.,$C$), thereby, making estimation of \mW\ scalable to larger itemsets. 

\bigskip
 The complete procedure for building a \recwalk\ model is given in Algorithm~\ref{alg:RecWalk}.\footnote{{Note that while in this paper the user-item interaction matrix \mR\ is assumed to be binary (\textit{implicit feedback} setting), 
 		\recwalk\ can be applied to 
 	 any non-negative feedback matrix that contains  e.g., `ratings', `click counts', etc. 
  }}

\bigskip
\begin{algorithm}
	\caption{\textsc{RecWalk Model}}
	\label{alg:RecWalk}
	\begin{algorithmic}
		\State \textbf{Input:} Input matrix \mR, \\ \qquad\ \ \ \  {\textit{parameters}:}  $\alpha$, $\gamma_1$, $\gamma_2$, $C$. \hfill
		\State \textbf{Output:} \recwalk\ transition probability matrix \mP.
		\ParFor { $i \in \set{I}$ }				
		\State Find the $C$ nearest neighbors of item $i$ and form $\mN_i$
		\begin{displaymath}
		\qquad\ \ \MINone{\vx \in \mathfrak{R}^C}{\frac{1}{2}\lVert\vr_i - \mN_i\vx \rVert^2_2 + \gamma_1 \lVert \vx\rVert_1 + \frac{1}{2}\gamma_2 \lVert \vx \rVert^2_2}
		{\vx\geq \mzero} 
		\end{displaymath}
		\State Fill the corresponding elements in the $i-$th column of \mW. 
		\EndParFor		
		\State 	$\mM_\set{I} \gets \frac{1}{\lVert \mW \rVert_\infty} \mW + \Diag(\ones-\frac{1}{\lVert \mW \rVert_\infty}\mW\ones)$	
		\State $
		\mA_\graph{G} \gets \bmat{ \mzero 		   & \mR \\ 
			\mR^\transpose & \mzero  }$
		\State $ \mP \gets \alpha\Diag(\mA_\graph{G}\ones)^{-1}\mA_\graph{G} + (1-\alpha) \bmat{\mI & \mathbf{0}\\ \mathbf{0}& \mM_\set{I}}$ 
	\end{algorithmic}
\end{algorithm}

\subsection{Recommendation Strategies}
Having defined the \recwalk\ transition probability matrix we can produce recommendations by exploiting the information captured in the successive landing probability distributions of a walk initialized in a user-specific way. Here, we will consider two recommendation strategies; namely
\begin{description}
	\item[\kstep:] The recommendation score of user $u$ for item $i$ is defined to be the probability the random walker lands on node $i$ after $K$ steps, given that she started on node $u$. Therefore, the recommendation score for item $i$ is given by the corresponding elements of   
	\begin{equation}
	\vpi_u^\transpose \mathdef \ve_u^\transpose \mP^K
	\end{equation} 
	where $\ve_u \in \mathfrak{R}^{U+I}$ is a vector that contains the element 1 on the position that corresponds to user $u$ and zeros elsewhere. The computation of the recommendations is presented in Algorithm~\ref{alg:RecWalkSstep} and it entails
	\begin{displaymath}
	\Theta(K\operatorname{nnz}(\mP))
	\end{displaymath}
	operations, where $\operatorname{nnz}(\mP)$ is the number of nonzero elements in \mP. 
	\item[\pr:] The recommendation score of user $u$ for item $i$ is defined to be the element that corresponds to item $i$ in the limiting distribution of a random walk with restarts on \mP, with damping factor $\eta$ and teleportation distribution $\ve_u$:  
	\begin{equation}
	\vpi_u^\transpose \mathdef \lim\limits_{K\to \infty}\ve_u^\transpose(\eta\mP + (1-\eta)\ones\ve_u^\transpose)^K. 
	\end{equation} 
	This can be computed using the power method as in Algorithm~\ref{alg:RecWalkPR}. Producing recommendations for a user involves 
	\begin{displaymath}
	\Theta((\log\epsilon/\log\eta)\operatorname{nnz}(\mathbf{P}))
	\end{displaymath}
	floating point operations for convergence up to a tolerance $\epsilon$~\cite{meyer2012stochastic}. 
\end{description} 

\begin{algorithm}
	\caption{\textsc{RecWalk}\textsuperscript{K-step}}
	\label{alg:RecWalkSstep}
	\begin{algorithmic}
		\State \textbf{Input:} \recwalk\ model \mP, user $u\in\set{U}$. 
		\State \textbf{Output:} Recommendation vector $\vpi_u$. 
		\State $\vpi_u^\transpose \gets \ve_u^\transpose$
		\For {$k \in 1,\dots,K $} 
		\State $\vpi_u^\transpose \gets \vpi_u^\transpose \mP$ 
		\EndFor
	\end{algorithmic}
\end{algorithm}
\begin{algorithm}
	\caption{\textsc{RecWalk}\textsuperscript{PR}}
	\label{alg:RecWalkPR}
	\begin{algorithmic}
		\renewcommand{\algorithmicrequire}{\textbf{Input:}}
		\renewcommand{\algorithmicensure}{\textbf{Output:}}
		\State \textbf{Input:} \recwalk\ model \mP, user $u\in\set{U}$, damping factor $\eta$. 
		\State \textbf{Output:} Recommendation vector $\vpi_u$.
		\State $\vx_{(0)}^{\transpose} \gets \ve_u^\transpose$ 
		\State $k\gets0$
		\Repeat
		\State $k \gets k+1$
		\State $\vx_{(k)}^{\transpose} \gets   \eta\vx_{(k-1)}^{\transpose} \mP + (1-\eta)\ve_u^\transpose$
		\State Normalize $\vx_{(k)}^{\transpose}$
		\Until{$\lVert \vx_{(k)}^{\transpose} - \vx_{(k-1)}^{\transpose}\rVert_1 < \mathit{tol}$}
		\State $\vpi_u \gets \vx_{(k)}$
	\end{algorithmic}
\end{algorithm}

\section{Theoretical Properties} 
As we will show in this section, a key property of the \recwalk\ model is that for small values of parameter $\alpha$  the \recwalk\ chain is going to be \textit{nearly uncoupled} into a large number of blocks, thereby allowing the random walk process dynamics towards equilibrium to disentangle into a \textit{slow-mixing} and a \textit{fast-mixing} component. This implies personalized landing probabilities even when the number of steps gets large.

\subsection{Nearly Uncoupled Markov Chains}
A nearly uncoupled Markov chain is a discrete time chain whose transition probability matrix is almost block diagonal~\cite{courtois1977decomposability,stewart1991sensitivity}. Formally, let $\mZ \in \mathfrak{R}^{n \times n}$ be the transition probability matrix of an irreducible and aperiodic Markov chain. Matrix $\mZ$ can always be written as
\begin{displaymath}
\mZ=\mZ^\star+\varepsilon \mC, 
\end{displaymath}  
where $\mZ^{\star}$ is a block-diagonal matrix of order $n$,  given by 
\begin{displaymath}
\mZ^\star \mathdef \pmat{\mZ^\star_{11} & \mzero & \cdots & \mzero \\
\mzero &   \mZ^\star_{22} & \ddots & \vdots \\
\vdots & \ddots & \ddots & \mzero \\
\mzero & \cdots & \mzero & \mZ^\star_{LL} } 
\end{displaymath}
and matrices $\mZ^\star_{II}$  
are irreducible stochastic matrices of order $n(I)$. 
Hence,   
\begin{displaymath}
n=\sum_{I=1}^{L}{n(I)}
\end{displaymath}  
and because both  $\mZ$ and  $\mZ^{\star}$ are stochastic, the row-sums of $\mC$ are zero.

Matrix $\mC$  and the non-negative real 
number $\varepsilon$ are selected such that for all rows it holds  
\begin{align}
\varepsilon \sum_{J\ne I}\sum_{l=1}^{n(J)} [\mC_{IJ}]_{ml} &=  \sum_{J\ne I}\sum_{l=1}^{n(J)}[\mZ_{IJ}]_{ml}, \\
\varepsilon & \mathdef \max_{m_I} \sum_{J\ne I}\sum_{l=1}^{n(J)} [\mZ_{IJ}]_{ml}
\end{align}
where we use $[\mZ_{IJ}]_{ml}$ to denote the element at the intersection of the $m$-th  row and $l$-th column of the $\mZ_{IJ}$  submatrix of $\mZ$.

Parameter $\varepsilon $ is referred to as the \emph{maximum degree of coupling} between the blocks. When $\varepsilon $ is sufficiently small, the Markov chain with transition probability matrix $\mZ$  is called \emph{nearly uncoupled} into $L$ blocks~\cite{courtois1977decomposability}.

\subsection{Mixing properties of \recwalk} 
\label{subsec:mixingproperties}
When graph $\graph{G}$ is connected the discrete time Markov chain defined by \mP\ is ergodic (see Appendix~\ref{Appendix:Ergodicity} for a proof). Thus, as the number of steps $K$ increases the landing probabilities of \recwalk\ will converge to a limiting distribution.  It is well known (see e.g.,\cite{stewart:1994introduction}) that the rate of convergence to this distribution depends on the modulus of {the subdominant eigenvalue of the transition probability matrix of the walk, which we denote $\lvert\lambda_2(\mP)\rvert$}. In particular, the asymptotic rate of convergence to the limiting distribution is the rate at which 
\begin{displaymath}
\lvert\lambda_2(\mP)\rvert^k \to 0.
\end{displaymath} Intuitively, the smaller $|\lambda_2(\mP)|$ is, the sooner the  landing probability distributions will start yielding recommendation vectors that are ``non-personalized,'' in the sense that they are similar for all users irrespectively of the items with which they have interacted.

The following theorem sheds more light to the spectral properties of matrix $\mathbf{P}$.

\begin{theorem}
	Let \mP\ be the \recwalk\ transition probability matrix with $\alpha \in (0,1)$ defined over a connected user-item network $\mathcal{G}$, and also let $\lambda(\mP)$ be the set of the eigenvalues of \mP. Irrespectively of the item model used to define matrix $\mM_\set{I}$ it holds 
	\begin{enumerate} \itemsep 1pt
		\item[(a)] $1-2\alpha \in \lambda(\mP)$
		\item[(b)] when $\alpha$ is small enough the Markov chain with transition probability matrix \mP\ will be nearly uncoupled into at least $U+1$ blocks. 
	\end{enumerate}
	\label{thm:EigMatP}
\end{theorem}

\begin{proof}
	When $\mathcal{G}$ is connected,  the stochastic matrix $\mH$ is irreducible~\cite{langville2011google}. Moreover, since the graph is bipartite a simple random walk on $\graph{G}$ results in a periodic Markov chain with period $d=2$. Therefore, from the Perron-Frobenius theorem~\cite{grimmett2001probability} we get that 
	\begin{align}
	\lambda_1(\mH)  &=   1,  \nonumber\\
	\lambda_2(\mH)  &= e^{2 i\pi /d} = e^{i\pi} = -1. \nonumber
	\end{align}  
	The so-called Perron eigenvalue $\lambda_1(\mH)$ is associated with the right eigenvector $\boldsymbol{1}$; whereas eigenvalue $\lambda_2(\mH)$ with a right eigenvector which we denote $\mathbf{v}$. 
	
	The special structure of \mH\ makes it easy to guess the form of the eigenvector $\mathbf{v}$ as well as to verify that it actually denotes an eigenvector of matrix \mM\ too. In particular, we have 
	
	\begin{displaymath}
	\mathbf{v} \triangleq  [\overbrace{1 \quad 1 \quad 1 \quad \cdots \quad 1}^{\textrm{$|\mathcal{U}|$ user nodes}} \quad \overbrace{-1 \quad -1 \quad \cdots \quad -1}^{\textrm{$|\mathcal{I}|$ item nodes}}].
	\label{Eq:v}
	\end{displaymath} 
	It is easy to see that $\mathbf{v}$ is indeed an eigenvector of both matrices $\mH$ and $\mM$. Specifically, we have  
	\begin{equation}
	\mH\mathbf{v}  =  \begin{bmatrix}
	\mathbf{0} & \mathbf{H}_{12} \\
	\mathbf{H}_{21} & \mathbf{0} 
	\end{bmatrix} \begin{bmatrix}
	\mathbf{1}_{U}\\
	-\mathbf{1}_{I}
	\end{bmatrix}  =  \begin{bmatrix}
	-\mathbf{1}_{U}\\
	\mathbf{1}_{I}
	\end{bmatrix} = - \mathbf{v} 
	\label{eigenpairH}
	\end{equation}
	from which we get that $(-1,\mathbf{v})$ is an eigenpair of matrix $\mH$; and 
	\begin{equation}
	\mM\mathbf{v} 
	= \begin{bmatrix}
	\mathbf{I} & \mathbf{0} \\
	\mathbf{0} & \mathbf{M}_{\set{I}} 
	\end{bmatrix} 
	\begin{bmatrix}
	\mathbf{1}_{U}\\
	-\mathbf{1}_{I}
	\end{bmatrix}  
	 =  \begin{bmatrix}
	\mathbf{1}_{U}\\
	-\mathbf{1}_{I}
	\end{bmatrix} = \mathbf{v}
	\label{eigenpairM}
	\end{equation}
	which implies that $(1,\mathbf{v})$ is an eigenpair of matrix $\mM$.
	
	Now consider a non-singular matrix, 
	\begin{equation}
	\mQ \triangleq \begin{bmatrix}
	\mathbf{1} & \mathbf{v} & \mX\end{bmatrix},
	\end{equation} 
	which contains in its first two columns the  eigenvectors \ones\ and $\mathbf{v}$. Also let 
	\begin{equation}
	\mathbf{Q}^{-1} \triangleq \begin{bmatrix}
	\vy^\transpose_1 \\[0.07cm]
	\vy^\transpose_2 \\[0.07cm]
	\mY^\transpose
	\end{bmatrix}. 
	\end{equation}
	By definition it holds $\mathbf{Q}^{-1}\mQ  = \mathbf{I}$, which can be usefully written as 	
	\begin{eqnarray}
	\begin{bmatrix}
	\vy^\transpose_1 \mathbf{1} & \vy^\transpose_1\mathbf{v} & \vy^\transpose_1\mX \\[0.07cm]
	\vy^\transpose_2\mathbf{1} & \vy^\transpose_2\mathbf{v} & \vy^\transpose_2\mX \\[0.07cm]
	\mY^\transpose \mathbf{1} & \mY^\transpose\mathbf{v} & \mY^\transpose\mX \\
	\end{bmatrix} & = & \begin{bmatrix}
	1 & 0 & \mathbf{0} \\[0.07cm]
	0 & 1 & \mathbf{0} \\[0.07cm]
	\mathbf{0} & \mathbf{0} & \mathbf{I} \\
	\end{bmatrix}.
	\label{identity}
	\end{eqnarray}

	Now, if we consider the similarity transformation of the \recwalk\ transition probability matrix, $\mQ^{-1}\mP\mQ$, also taking into consideration the relations \eqref{eigenpairH}, \eqref{eigenpairM} and the identities (\ref{identity}), we have 
	\begin{eqnarray}
	\footnotesize		
	\label{Eq:SimilarityTransfomationG}
	\mQ^{-1}\mP\mQ 
	& = & \alpha\mQ^{-1}\mH\mQ + (1-\alpha) \mQ^{-1}\mM\mQ  \nonumber \\[0.1cm]
	& = & 	\alpha\begin{bmatrix}
	\mY^\transpose_1\mathbf{1} & (-1)\mY^\transpose_1\mathbf{v} & \mY^\transpose_1\mH\mX \\[0.07cm]
	\mY^\transpose_2\mathbf{1} & (-1)\mY^\transpose_2\mathbf{v} & \mY^\transpose_2\mH\mX \\[0.07cm]
	\mY^\transpose\mathbf{1} & (-1)\mY^\transpose\mathbf{v} & \mY^\transpose\mH\mX \\
	\end{bmatrix}  + 
	(1-\alpha) 	\begin{bmatrix}
	\mY^\transpose_1\mathbf{1} & \mY^\transpose_1\mathbf{v} & \mY^\transpose_1\mM\mX \\[0.07cm]
	\mY^\transpose_2\mathbf{1} & \mY^\transpose_2\mathbf{v} & \mY^\transpose_2\mM\mX \\[0.07cm]
	\mY^\transpose\mathbf{1} & \mY^\transpose\mathbf{v} & \mY^\transpose\mM\mX \\
	\end{bmatrix} 	\nonumber \\[0.1cm]
	& = & \begin{bmatrix}
	\alpha & 0 & \alpha\mY^\transpose_1\mH\mX \\[0.07cm]
	0 & -\alpha & \alpha\mY^\transpose_2\mH\mX \\[0.07cm]
	\mathbf{0} & \mathbf{0} & \alpha\mY^\transpose\mH\mX \\
	\end{bmatrix} + 
	\begin{bmatrix}
	1-\alpha & 0 & (1-\alpha)\mY^\transpose_1\mM\mX \\[0.07cm]
	0 & 1-\alpha & (1-\alpha)\mY^\transpose_2\mM\mX \\[0.07cm]
	\mathbf{0} & \mathbf{0} & (1-\alpha)\mY^\transpose\mM\mX
	\end{bmatrix} \nonumber\\[0.1cm]
	& = & \begin{bmatrix}
	1 & 0 & \alpha\vy^\transpose_1\mH\mX+(1-\alpha)\vy^\transpose_1\mM\mX \\[0.07cm]
	0 & 1-2\alpha & \alpha\vy^\transpose_2\mH\mX+(1-\alpha)\vy^\transpose_2\mM\mX \\[0.07cm]
	\mathbf{0} & \mathbf{0} & \alpha\mY^\transpose \mH\mX+(1-\alpha)\mY^\transpose\mM\mX 
	\end{bmatrix} \label{eq:therest}
	\end{eqnarray}	
	
	Thus, matrix \mP\ is similar to a block upper triangular matrix, the eigenvalues of which are the eigenvalues of its diagonal blocks. From that, we directly establish that $1-2\alpha$ is an eigenvalue of the \recwalk\ transition matrix $\mP$, and the first part of the theorem is proved. 
	
	To prove the second part it suffices to show that there exists a partition of the state space of the \recwalk\ chain into blocks, such that the maximum probability of leaving a block upon a single transition is upper-bounded by $\alpha$~\cite{stewart1991sensitivity}.  In particular, consider the partition  
	\begin{equation}
	\set{A} \mathdef  \{\{u_1\}, \{u_2\}, \dots, \{u_U\}, \set{I}\}.  
	\end{equation} 
	By definition, in the \recwalk\ model the probability of leaving $u$ is equal to $\alpha$, for all $u\in\set{U}$. Concretely, 
	\begin{equation}
	\Pr\{\text{jump from}\ u\in\set{U}\ \text{to any}\ j\neq u \} = \sum_{j \neq u} P_{uj}  =\sum_{j \neq u} \alpha H_{uj} = \alpha. 
	\end{equation}
	Similarly, the probability of leaving block $\set{I}$ upon a transition is 
	\begin{equation}
	\Pr\{\text{jump from}\ i \in \set{I}\ \text{to any}\ \ell\notin \set{I} \}  = \sum_{\ell \notin \set{I}} P_{i\ell}  =  \sum_{\ell \notin \set{I}} \alpha H_{i\ell} = \alpha. 
	\end{equation}  
	Therefore, the \recwalk\ chain can always be decomposed according to $\set{A}$  such that the maximum degree of coupling between the involved blocks is exactly equal to $\alpha$. Hence, choosing $\alpha$ to be sufficiently small ensures that the chain will be nearly uncoupled into (at least) $U+1$ blocks. 
\end{proof}

\bigskip
Theorem~\ref{thm:EigMatP} asserts that the proposed random walk construction ensures the existence of an eigenvalue equal to $1-2\alpha$. This means that the modulus of the eigenvalue that determines the rate of convergence to the limiting distribution will be at least $1-2\alpha$. Hence, choosing $\alpha$ allows us to ensure that the \recwalk\ process will converge as slow as we need to increase the number of landing probability distributions that can still serve as personalized recommendation vectors in our model---\textit{irrespectively of the particular user-item network or the chosen item model upon which it is built.} 

Moreover note that the spectral fingerprint of nearly uncoupled Markov chains is the existence of a set of subdominant eigenvalues that are relatively close (but not equal) to 1~\cite{courtois1977decomposability}. In our case, for small values of $\alpha$ these eigenvalues are expected to be clustered close to the value $1-\alpha$ (cf.~\eqref{Eq:SimilarityTransfomationG}).  The number of these eigenvalues depicts the number of blocks of states into which the chain is nearly uncoupled. Therefore, subject to $\alpha$ being small the \recwalk\ chain  will have at least $U+1$ eigenvalues clustered near the value 1, and it can be shown (see e.g.,\cite{stewart1991sensitivity}) that matrix \mP\ can be expressed  as  
\begin{equation}
\mP =  \ones\vpi^\transpose + \mT_\mathrm{slow} + \mT_\mathrm{fast}
\end{equation}
where $\vpi^\transpose$ is the stationary distribution of the walk,  $\mT_\mathrm{slow}$ is a \textit{slow transient} component, and $\mT_\mathrm{fast}$ is a \textit{fast transient} component. As $K$ gets large the fast transient term will diminish while the elements of the slow transient term will remain large enough to ensure that the recommendation vectors are not completely  dominated by $\ones\vpi^\transpose$. Of course, as $K$ gets larger and larger the relative influence of the first term will become stronger and stronger, up to the point where each user is assigned the exact same recommendation vector $\vpi^\transpose$; however, this outcome will be delayed by the existence of the slow transient term $\mT_\mathrm{slow}$. Note that in a simple random walk on \mW\ such time-scale dissociation of the stochastic dynamics of the walk is typically absent; and certainly it cannot be guaranteed in advance. On the contrary, the proposed random walk construction in \recwalk\ provides a clear mechanism to ensure such property, and as we will see in the experimental section of this paper this property alone can lead to significant improvements in top-$n$ recommendation quality compared to what one would get by using the item model directly.   

\bigskip
\subsection{Coverage}
\label{subsec:coverage}
It is often useful to be able to ensure that in principle all items can be recommended to the users. However, due to the sparsity of the user-item feedback matrix, such requirement is many times difficult to satisfy or comes at the expense of accuracy. Moreover, the opaque nature of many of the state-of-the-art approaches for top-n recommendation makes it impossible to guarantee such property in advance. 

In case of \pr\ this requirement is easily guaranteed. Before we proceed to the proof of this property, we list the following useful Lemma from~\cite{nikolakopoulos2015random}. 

\begin{lemma}[\cite{nikolakopoulos2015random}]
	\label{lemma}
	Let $\mathbf{A}$ be a primitive stochastic matrix and $\mB_1,\mB_2,\dots,\mB_n$ stochastic matrices; then matrix
	\begin{displaymath}
	\alpha \mA+\beta_1\mB_1+\dots+\beta_n\mB_n
	\end{displaymath} where $\alpha>0$ and $\beta_1,\dots,\beta_n\geq0$ such that \begin{displaymath}
	\alpha+\beta_1+\dots+\beta_n=1
	\end{displaymath} is a primitive stochastic matrix.
\end{lemma} 
\begin{proposition}[\pr\ Coverage Conditions]
	When $\graph{G}$ is connected the personalized recommendation vectors produced by \pr\ with $\alpha,\eta \in (0,1)$ are guaranteed to provide full itemspace coverage for every user in the system, i.e.,  
	\begin{displaymath}
	[\vpi_u]_j > 0, \qquad \forall j \in \set{I}, \quad \forall u \in \set{U}.
	\end{displaymath}
\end{proposition}
\begin{proof}
	Notice that $\vpi_u$ in case of \pr\ is defined as the limiting distribution of a Markov chain with transition probability matrix 
	\begin{equation}
	\label{eq:matrixforproof}
	\eta\mP + (1-\eta)\ones\ve_u^\transpose.
	\end{equation} 
	Under our assumptions, the \recwalk\ chain is ergodic (see Proposition~\ref{Prop:ergodicity} in Appendix) and hence, matrix \mP\ is a primitive stochastic matrix~\cite{meyer2000matrix}. This makes stochastic matrix \eqref{eq:matrixforproof}, a linear combination of a primitive matrix, and the rank-one stochastic matrix, $\ones\ve_u^\transpose$; thus, Lemma~\ref{lemma} applies, and we get that \eqref{eq:matrixforproof} is a primitive stochastic matrix as well. 
	
	Therefore, the recommendation vector $\vpi_u$ produced by \pr\ will coincide with the unique dominant left eigenvector of the primitive stochastic matrix \eqref{eq:matrixforproof} (up to rescaling of its entries to sum to 1). However, from the Perron Frobenius theorem, we know that every element of the dominant left eigenvector of primitive matrices is necessarily strictly positive (see e.g.,~\cite{meyer2000matrix}). Thus, when $\graph{G}$ is connected and $\alpha,\eta \in (0,1)$, the recommendation vector $\vpi_u$ produced by \pr\ is guaranteed to assign positive scores to every item in the system, i.e,   
	\begin{displaymath}
	[\vpi_u]_j > 0, \qquad \forall j \in \set{I}, \quad \forall u \in \set{U}.
	\end{displaymath}
	as needed; and the proof is complete.
\end{proof}

\kstep\ can guarantee item-space coverage too, provided that $K$ is sufficiently large for the walk to discover the complete set of items. This is a direct consequence of the primitivity of \mP. Bounding the necessary number of steps $K$ to ensure coverage would require analysis of relevant combinatorial invariants of the particular user-item network and item-model under consideration, and thus, goes beyond the scope of this work. In practice, we find that even modest values of $K$ are enough to ensure such property for every dataset we experimented with. Importantly, we also find that for \recwalk\ item-space coverage \textit{does not come at the expense of accuracy}, as it is the case for other item-based recommendation algorithms.

\section{Experimental Setting}
\label{Sec:Experimental_Setting}

\subsection{Datasets} 
Our qualitative evaluation is based on six real-world publicly available datasets; namely (i) the \textit{movielens} dataset which contains the ratings of  users for movies and it has been used extensively for the evaluation of top-$n$ recommendation methods; (ii) the \textit{yahoo} dataset which is a subset of the Yahoo!R2Music dataset (obtained from~\cite{nikolakopoulos2014use}) containing the ratings of users for songs; (iii) the \textit{electronics} dataset~\cite{He:2016:UDM:2872427.2883037} which contains review scores of electronic products in Amazon (we keep users that have rated at least 10 items, and items that have been rated by at least 20 users); (iv) the \textit{movies\&tv} dataset~\cite{He:2016:UDM:2872427.2883037} which contains review scores of movies and tv shows in Amazon (we keep users that have scores at least 15 items, and items that have been scored by at least 30 users), (vi) \textit{books} datasets~\cite{He:2016:UDM:2872427.2883037}, which contain review scores of books in {Amazon} (we keep users that have provided review scores for at least 15 books, and books that have been scored by at least 30 users); and (vi) the \textit{pinterest} dataset (obtained from~\cite{ncf}) which captures the interactions of users regarding images where each interaction denotes whether the user has ``pinned'' the image to her own board. Basic statistics about the datasets can be found in Table~\ref{tab:Datasets}.

\begin{table}[h]
	\caption{ Statistics of datasets used to assess recommendation accuracy. }
	\centering
	\small
	\begin{tabular} {lrrrr r  }
		\toprule
		\toprule
		Name                & \#users & \#items & \#interactions & 
		& density \\
		\midrule
		\myrowcolour%
		\textit{electronics} & 3,769  & 1,598   & 66,060   & 
		& 0.0109  \\
		\textit{movielens}  & 6,040   & 3,706   & 1,000,029  & 
		& 0.0447  \\
		
		\myrowcolour%
		\textit{yahoo}      & 7,307   & 3,312   & 404,745   & 
		& 0.0167  \\

		\textit{movies\&tv}\phantom{asdf asdf} & 10,039  & 5,400   & 437,763   & 
		& 0.0081  \\
		\myrowcolour%
		\textit{books}      & 43,550  & 24,811  & 1,777,072  & 
		& 0.0016  \\
		\textit{pinterest}   & 55,187  & 9,916  & 1,463,581  & 
		& 0.0027  \\
		\bottomrule
		\bottomrule
	\end{tabular}
	\label{tab:Datasets}
\end{table}

 Furthermore, to assess the scalability of our method we make use of several larger datasets from the collection~\cite{He:2016:UDM:2872427.2883037}. Their characteristics can be found in Appendix~\ref{Appendix:LargeItemspaces}.

\subsection{Evaluation Methodology and Metrics}
\label{subsec:evaluationprotocol}
To  evaluate  the  top-$n$ recommendation performance, we adopted the widely used \textit{leave-one-out} evaluation protocol~\cite{rendle2009bpr,ning2011slim,ncf,nikolakopoulos2019eigenrec,koren2010factor}. In particular, for each user we randomly select one item she liked\footnote{When rating information is available in the original data, the per-user target item is randomly sampled  among the highest rated items of each particular user in order to ensure that it indeed denotes an item that the user liked. Such approach is in accordance with the methodology described in the seminal papers~\cite{Cremonesi:2010:PRA:1864708.1864721,koren2010factor}; in our setting, however, all users are equally represented in the testset. }
and we create a test set $\set{T}$. The rest of the dataset is used for training the models. For model selection we repeat the same procedure on the training data and we create a validation set $\set{V}$;  and for each method considered we explore the hyperparameter space to find the model that yields the best performance in recommending the items in $\set{V}$, and then we evaluate its out-of-sample performance based on the held-out items in $\set{T}$. For the evaluation we consider for each user her corresponding test item alongside 999 randomly sampled unseen items and we rank the 1000 item lists based on the recommendation scores produced by each method. During training of all competing methods we consider only binary feedback.

The evaluation of the top-$n$ recommendation performance is based on three widely used ranking-based metrics; namely the \textit{hit ratio} (HR@$n$), the \textit{average reciprocal hit rank} (ARHR@$n$), and the truncated \textit{normalized discounted cumulative gain} (NDCG@$n$) over the set of users (for a detailed definition we refer the reader to e.g.,~\cite{vaedae,nikolakopoulos2019eigenrec}). For each user, all metrics compare the predicted rank of the held-out item with the ideal recommendation vector which ranks the held-out item first among the items in each user's test-set list. For all competing methods we get the predicted rank by sorting the recommendation scores that correspond to the items included in each user's test-set list. While HR@$n$ gives a perfect score if the held-out item is ranked within the first $n$, ARHR@$n$ and NDCG@$n$ use a monotonically increasing discount to emphasize the importance of the actual position of the held-out item in the top-$n$ recommendation list. 

\section{Experimental Results} 
\label{Sec:Experimental_Results}
\label{sec:Empirical}

\subsection{Effect of Parameter $\alpha$}
The theoretical analysis of our method suggests that parameter $\alpha$ controls the convergence properties of the \recwalk\ process and it can be chosen to enable a time-scale dissociation of the stochastic dynamics towards equilibrium that ensures a larger number of personalized landing distributions. Here we verify experimentally the predicted properties and we evaluate their effect on the recommendation quality of the $K$-step landing probabilities. 

We build the item model \mW\ that yields the best performance on the validation set and we use it to create matrix $\mM_\set{I}$ as in \eqref{def:M_I}. We then build the \recwalk\ model as in Algorithm~\ref{alg:RecWalk}, we run it for different values of $\alpha$ ranging from 0.005 to 0.5, and we report: (i) the performance in terms of average NDCG$@n$ (for $n=10$) across all users for values of steps $K$ up to 30 (Fig.\ref{fig:alphatests}-A); (ii) the spectra of the corresponding transition probability matrices \mP\ (Fig.\ref{fig:alphatests}-B); (iii) the peak performance per parameter $\alpha$ along with the step for which it was achieved (Fig.\ref{fig:alphatests}-C); and (iv) the performance of \recwalk\ with respect to using the base model \mW\ directly (Fig.\ref{fig:alphatests}-D).

\begin{figure*}[h!]
	\centering
	\includegraphics[width=0.99\linewidth]{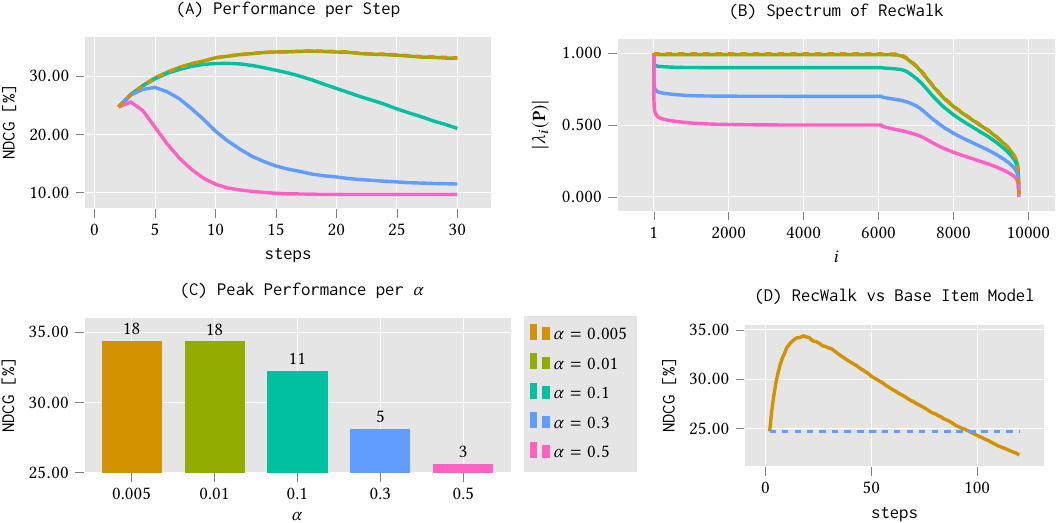}
	\caption{Fig (A) reports the performance of \recwalk\ in terms of NDCG@10 as a function of the steps for different values of the parameter $\alpha$. Fig (B) plots the spectrum of the \recwalk\ transition probability matrix for different values of $\alpha$. Fig (C) reports the peak performance for the different values of $\alpha$, as well as the number of steps for which it is achieved (on top of each bar). Fig (D) reports \recwalk\ performance (with $\alpha$ fixed at 0.005) for different number of steps compared to the performance one would get using the item model directly (blue dashed line). }
	\label{fig:alphatests}
\end{figure*}

We find that as the value of $\alpha$ gets smaller the top-$n$ recommendation quality increases and stabilizes for $\alpha<0.01$. Similarly the number of steps that yield the best performance increase (see Fig.\ref{fig:alphatests}-C). The spectrum of the corresponding transition probability matrices reflects the theoretically predicted properties of the model. Indeed for very small values of $\alpha$ we observe that the subdominant eigenvalues cluster near the value $1$, thus forming the slowly varying component of matrix $\mP$, which ensures that the successive landing probabilities of the  random walk are  influenced by the initial state for longer. Furthermore, we find that the proposed methodology entails a significant increase in performance with respect to what one would get by using the proposed base item-model directly. In particular the recommendation performance of the $K$-step landing probabilities of \recwalk\ overpasses the performance of the base model (see Fig.\ref{fig:alphatests}-D) for a wide range of steps up to a maximum increase of 39.09\%  for $K=18$ steps. {Finally, notice that when the number of steps becomes very large the recommendation performance decays, as every $\vpi_u$ starts to converge to the same stationary distribution $\vpi$. In other words, as $K$ gets larger and larger similar recommendations will be produced for every user in the system.   Due to the mixing properties of the \recwalk\ chain, however, one can ensure that such outcome will not occur for small $K$---thereby allowing \recwalk\ to effectively explore the underlying network \textit{without losing focus of the user-specific initialization}.}

Our results suggest that the mixing properties of the \recwalk\ chain are indeed intertwined with the top-$n$ recommendation quality and can lead to a significant boost in performance. This justifies the intuition that motivated the particular design of the walk. 
	{For additional experiments regarding the effects of sparsity on the performance of \recwalk\ see Appendix~\ref{Appendix:Sparsity}.}

\subsection{RecWalk as a Framework}
In the definition of the inter-item transition probability matrix we proposed a particular strategy for constructing matrix \mW\ that was designed to promote locality on the direct item-to-item transitions  while being  easy to compute. Instead of this particular matrix \mW\ one could use any model that captures inter-item relations. But does our approach offer any benefit with respect to performing simple random walks on the corresponding item-model or to simply using the item model directly? 

Here we explore this question by empirically evaluating two commonly used item models. Namely: 
\begin{enumerate}
	\item a cosine similarity model $\mW_{\textsc{cos}}$ defined such that its $ij$-th element is given by $\vr_i^\transpose \vr_j / (\lVert \vr_i \rVert \lVert \vr_j \rVert)$; 
	\item a \text{SLIM} model which learns a matrix $\mW_{\textsc{slim}}$ by solving an $\ell_1,\ell_2$ regularized optimization problem (see~\cite{ning2011slim} for details).
\end{enumerate} 
We consider the respective base models alongside six approaches based on random walks; namely
\begin{enumerate}
	\item[(i)] \text{SRW}, which recommends using the $K$-step distribution of a simple random walk on \mW\ with transition probability matrix $\mS$ initialized with $\boldsymbol{\phi}_u^\transpose$ as in \eqref{model:srw};
	\item[(ii)] \text{PR}, which produces recommendations based on \mS\ as in \eqref{model:simplepr};
	\item[(iii)] \kstep;
	\item[(iv)]  \pr;
	\item[(v)] \text{\recwalk$[\mM_\set{I}]$\textsuperscript{K-step}}, which produces recommendations as in \eqref{model:srw} but using the \recwalk\ inter-item transition probability matrix introduced in \eqref{def:M_I} instead of \mS;
	\item[(vi)] \text{\recwalk$[\mM_\set{I}]$\textsuperscript{PR}}, which produces recommendations as in \eqref{model:simplepr} using RecWalk's $\mM_\set{I}$ instead of \mS. 
\end{enumerate}

\begin{table}[h]
	\caption{Top-n Recommendation Quality under Different Random Walk Constructions}
	\label{tab:FrameWorkResults}
	\small
	\centering
	\begin{threeparttable} 
		\begin{tabular}{l cc}
			\toprule%
			\toprule%
			\centering%
			Method
			& \text{COS}
			& \text{SLIM}
			\\	
			\midrule	
			\myrowcolour%
			Base model \phantom{asdfa asdf asdf adf asdg asdf asdf asdfas dfas} & 17.61 & 27.28 \\
			\midrule
			\text{SRW} & 17.82 & 25.37 \\
			\myrowcolour%
			\text{PR} & 18.11 & 25.37 \\
			\kstep\  & 20.52 & 31.87
			\\
			
			\myrowcolour
			\pr\  & 20.33 & 31.80
			\\
			\text{\recwalk$[\mM_\set{I}]$\textsuperscript{K-step}}  & 17.85 & 31.41\\
			\myrowcolour%
			\text{\recwalk$[\mM_\set{I}]$\textsuperscript{PR}} & 20.27 & 31.78 \\
			\bottomrule
			\bottomrule
		\end{tabular}
		\begin{tablenotes}
			\footnotesize
			\item Hyperparameters: \text{SRW}: K$ \in \{1,\dots,50\}$; \text{PR}: $p \in \{0.1,\dots,0.9\}$;  
			\text{SLIM}: $\lambda,\beta \in \{0.1, 0.5, 1, 3, 5, 10, 20\}$; \recwalk: $\alpha=0.005$ and $K\in\{1,\dots,30\}$ for \kstep\  and $\eta\in\{0.1,0.2,\dots,0.9\}$ for  \pr.
		\end{tablenotes}
	\end{threeparttable}
\end{table}

We run all models on the \movielens\ dataset and in Table~\ref{tab:FrameWorkResults} we report their performance on the NDCG@$n$ metric.  We see that \kstep\  and \pr\    were able to boost the performance of both item models (up to +16.52\%  for COS and +16.82\% for SLIM) with the performance difference between the two variants being insignificant. Applying simple random walks (\text{SRW}) or random walks with restarts (PR) directly to the row-normalized version of the item graph does not perform well (+2.84\% in case of COS and -7\% in case of SLIM). In particular in the case of \text{SRW} we observed a rapid decay in performance (after the first step for SLIM and after the first few steps for COS); similarly in case of \text{PR} the best performance was obtained for very small values of $p$---essentially enforcing  the $K$-step landing probabilities after the first few steps to contribute negligibly to the production of the recommendation scores (cf.~\eqref{model:simplepr}). Using  \recwalk's inter-item transition probability matrix alone, on the other hand, performed very well especially when we use the SLIM model as a base.

To gain more insight into the observed differences in performance of the walks on the item graphs, we also plot the spectra of the transition probability matrices $\mM_\set{I}$, alongside the spectra of the respective matrices \mS\  (Fig.~\ref{fig:FrameworkOtherWalks}). We see that in case of \mS\ the magnitude of the eigenvalues cause the walks to mix very quickly.  In case of matrix $\mM_\set{I}$ on the other hand, the eigenvalues decay more gracefully and on the SLIM graph in particular, there appears to be a large number of eigenvalues near 1, which delay the convergence of the landing distributions towards equilibrium. This effect is not as pronounced in case of COS which is reflected in the small increase in performance in case of \text{\recwalk$[\mM_\set{I}]$\textsuperscript{K-step}}.

\begin{figure*}[!ht]
	\centering	 
	\includegraphics[width=0.9\linewidth]{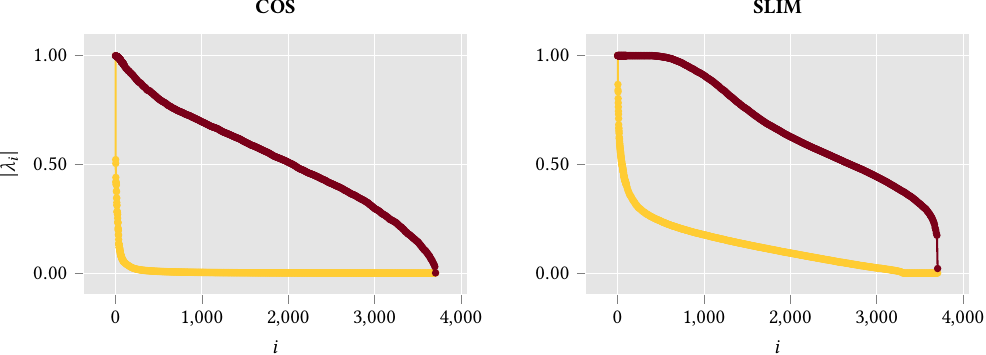}
	\caption{The figure plots the spectra of the transition probability matrices \mS\ (\textcolor{GopherGold}{\textbf{gold}} line) and $\mM_\set{I}$ (\textcolor{GopherMaroon}{\textbf{maroon}} line) defined using the corresponding item models.}
	\label{fig:FrameworkOtherWalks}
\end{figure*}

Again our experiments reveal a clear connection between the mixing properties of the walks and their potential in achieving good recommendation quality. Note also that the stochasticity adjustment strategy proposed in \eqref{def:M_I} seems to promote slow mixing in itself. However, using $\mM_\set{I}$ alone cannot in general guarantee this property irrespectively of the underlying item model whereas using the complete \recwalk\ model can give absolute control over convergence (as Theorem~\ref{thm:EigMatP} predicted).

\begin{table*}[h!]
	\caption{Top-$n$ recommendation quality of the competing approaches in terms of HR@10, ARHR@10 and NDCG@10.}
	\label{tab:Results}
	\centering
	\small
	\begin{threeparttable}
		\begin{tabular}{l c ccc c ccc}
			\toprule%
			\toprule%
			\centering%
			& \phantom{adas}
			& \multicolumn{3}{c}{{{\it movielens}}}
			& \phantom{adas}
			& \multicolumn{3}{c}{{{\it yahoo}}}
			\\
			\cmidrule[0.4pt](lr{0.125em}){3-5}%
			\cmidrule[0.4pt](lr{0.125em}){7-9}%
			Method    &      & \phantom{A}HR[\%]\phantom{R}                                      & ARHR[\%]            & NDCG[\%]          
			&  &
			\phantom{A}HR[\%]\phantom{R}               & ARHR[\%]                                    & NDCG[\%]             
			\\
			\midrule
			\myrowcolour%
			\puresvd\    &   & 44.14                                  & 19.33           &
			25.36           & 
			&
			38.68           & 18.30                                  & 22.62
			\\
			
			SLIM     &       & 46.34                                  & 21.39           & 27.28           &
			& 
			52.24           & 23.21                                  & 30.03
			\\
			
			\myrowcolour
			\eigenrec\  &    & 45.21                                  & 20.44           &
			26.35           & 
			&
			48.12           & 23.30                                  & 29.23
			\\
			
			\dae\     &      & 44.06                                  & 18.97           & 24.83           &
			& 
			45.37           & 21.46                                  & 27.07
			\\
			\myrowcolour
			\vae\    &       & 44.35                                  & 19.50           & 25.31           &
			& 
			45.09           & 21.22                                  & 26.80
			\\
			
			\apr\  &  & 42.45 & 18.69 & 24.24 & 
			&
			44.75 & 19.50 & 25.36 
			\\
			
			\myrowcolour
			\rpb\    &       & 34.87                                  & 15.02           & 19.66           &
			& 
			41.51           & 17.82                                  & 22.94
			\\
			
			\nais\    &      & 46.36                                  & 20.65           & 26.68           &
			&
			50.53           & 23.64                                  & 29.91
			\\

			\midrule
			\myrowcolour
			\kstep    &       &{50.28} & {27.20} & {33.13} &
			&
			\highest{55.02} & \highest{28.94} & \highest{35.10}  
			\\
			\pr     &      & \highest{52.52}                                  & \highest{27.74}           & \highest{33.57}           &
			&
			54.78           & 28.71                                  & 34.87
			\\
		\end{tabular}
		\begin{tabular}{l c ccc c ccc}
		\midrule
		\centering%
		& \phantom{adas}
		& \multicolumn{3}{c}{{{\it pinterest}}}
		& \phantom{adas}
		& \multicolumn{3}{c}{{{\it movies\&tv}}}
		\\
		\cmidrule[0.4pt](lr{0.125em}){3-5}%
		\cmidrule[0.4pt](lr{0.125em}){7-9}%
		Method    &      & \phantom{A}HR[\%]\phantom{R}                                      & ARHR[\%]            & NDCG[\%]                   
		&  &
		\phantom{A}HR[\%]\phantom{R}              & ARHR[\%]                                    & NDCG[\%]
		\\
		\midrule
		\myrowcolour%
		\puresvd\    &   & 30.97 & 11.85 & 16.30 & & 22.58                                  & 09.88           & 12.86
		\\
		
		SLIM     &      & 34.17 & 13.63 & 18.57 & 
		& 27.26                                  & 12.95           & 16.37
		\\
		
		\myrowcolour
		\eigenrec\  &    & 33.81 & 13.51 & 18.41  &
		& 25.22                                  & 11.44           & 14.66
		\\
		
		\dae\     &     & 35.03 & 13.79 & 18.77    &
		& 27.10                                  & 11.96           & 15.50 
		\\
		\myrowcolour
		\vae\    &      & 35.13 & 13.73 & 18.71   &
		& 26.72                                  & 12.05           & 15.40 
		\\
		
		\apr\  & & 33.93 & 13.11 & 17.94  &
		& 22.99 & 09.46 & 12.58
		\\
		
		\myrowcolour
		\rpb\    &     & 27.01	& 8.07 & 12.45 &
		& 26.52 & 11.79 & 15.22
		\\
		
		\nais\    &     & 34.06 & 12.95 & 17.82   &
		& 24.35                                  & 10.87           & 13.99 
		\\	
		
		\midrule
		\myrowcolour
		\kstep    &      & \highest{35.38} & \highest{14.07} & 18.95  &
		& \highest{28.34}  & 12.99 & \highest{16.51} 
		\\
		\pr     &     & 35.29 & \highest{14.07} & \highest{19.00}  &
		&  28.15           & \highest{13.13}           & \highest{16.51}
		\\
	\end{tabular}
		\begin{tabular}{l c ccc c ccc}
			\midrule%
			\centering%
			& \phantom{adas}
			& \multicolumn{3}{c}{{{\it books}}}
			& \phantom{adas}
			& \multicolumn{3}{c}{{{\it electronics}}}
			\\
			\cmidrule[0.4pt](lr{0.125em}){3-5}%
			\cmidrule[0.4pt](lr{0.125em}){7-9}%
			Method    &      & \phantom{A}HR[\%]\phantom{R}                                      & ARHR[\%]            & NDCG[\%]          
			&  &
			\phantom{A}HR[\%]\phantom{R}               & ARHR[\%]                                    & NDCG[\%]             
			\\
			\midrule
			\myrowcolour%
			\puresvd\       &
			& 46.84                                  & 25.84           & 30.81
			&	& 9.60 & 3.71 & 5.07 
			\\
			
			SLIM           & 
			& 56.72                                  & 34.50           & 39.67
			&	& 13.64 & 6.05 & 7.85 
			\\
			
			\myrowcolour
			\eigenrec\      &
			& 52.89                                  & 29.36           & 34.93
			&	& 11.83 & 4.68 & 6.34 
			\\
			
			\dae\           &
			& 54.66                                  & 29.75           & 35.60    
			&	& 13.40 & 5.10 & 7.18          \\
			\myrowcolour
			\vae\           &
			& 53.85                                  & 29.24           & 35.08
			&	& 13.43 & 5.56 & 7.32 
			\\
			
			\apr\   &
			& 51.12 & 24.96 & 31.10 
			&	& 11.91 & 4.37 & 6.11 
			\\
			
			\myrowcolour
			\rpb   &
			& 55.60 & 31.68  & 37.35 
			&	& 14.49	& 6.08 & 8.03
			
			\\
			
			\nais\          &
			& 51.18                                  & 28.36           & 33.65
			&	& 12.68 & 5.17 & 6.84 
			\\
			\midrule
			\myrowcolour
			\kstep           &
			& 	\highest{57.73}        & \highest{34.71} & \highest{40.05}
			&	& \highest{14.65} & \highest{6.25} & \highest{8.05} 
			\\
			
			\pr           &
			& \highest{57.73}                                 & 34.65           & 39.94
			&	& 14.25 & 6.18 & 7.96 
			\\
			\bottomrule
			\bottomrule
		\end{tabular}
		\begin{tablenotes}
			\scriptsize
			\item Hyperparameters:
			\recwalk: (fixed) $\alpha = 0.005$, $C = \{2.5\%, 5\%, 7.5\%, 10\%, 25\%\} \text{ of }\card{\set{I}}$,  $\gamma_1 \in \{1,3,5,10\},\gamma_2 \in \{0.1, 0.5, 1, 3, 5, 7, 9, 11, 15, 20\}$ and $K \in \{3, \dots, 20\}$; for \pr\ $\eta = \{0.05,\dots, 0.95\}$.  
			\puresvd: $f \in \{10,20,\dots,1000\}$. \rpb: $b \in \{0,0.05,\dots,1\}$. 
			\slim: We used the implementation~\cite{slim}, $\lambda,\beta \in \{0.1, 0.5, 1, 3, 5, 10, 20\}$.
			\eigenrec: $f \in \{10,20,\dots,1000\}, d \in \{-2,-1.95,\dots,2\}$.
			\dae - \vae: we used the hyperparameter tuning approach provided by the authors in their publicly available implementation; we considered both architectures proposed in~\cite{vaedae}; namely $[I-200-I]$ and $[I-600-200-600-I]$. \apr: We used the implementation provided by~\cite{APR}, considering the hyperparametric setting discussed therein; 
			\nais: We used the implementation provided by~\cite{NAIS}, considering the parametric ranges discussed in the paper, namely $\alpha,k \in \{8,16,32,64\}$, $\beta=0.5$ and regularization parameters $\lambda \in \{10^{-6},\dots,1\}$. 
		\end{tablenotes}
		\vspace*{-0.4cm}
	\end{threeparttable}
\end{table*}

\subsection{Performance against Competing Approaches}
\label{subsec:PerformancevsCompeting}
We evaluate the top-$n$ recommendation accuracy of  \kstep\  and \pr\  against eight competing approaches; namely 
\begin{enumerate}
	\item[i.] the well-known \puresvd\  method~\cite{Cremonesi:2010:PRA:1864708.1864721}, which produces recommendations based on the truncated SVD of \mR;
	\item[ii.] the well-known item-based method \slim~\cite{ning2011slim} which builds a sparse item model by solving an $\ell_1,\ell_2$-regularized optimization problem;
	\item[iii.] the random-walk approach \rpb~\cite{christoffel2015blockbusters}, which recommends based on a short-length walk on the user-item graph after rescaling the landing probabilities to compensate for the inherent  bias of the walk towards popular items;  
	\item[iv.] the \eigenrec\  method~\cite{nikolakopoulos2019eigenrec}, which builds a factored item model based on a scaled cosine similarity matrix;  
	\item[v-vi.] \vae\ and \dae\ ~\cite{vaedae} which extend variational and denoising autoencoders to collaborative filtering using a multinomial likelihood and were shown to achieve state-of-the-art recommendation quality,  outperforming several other deep-network-based approaches;
	\item[vii.]  \apr\ ~\cite{APR}, which extends the popular pairwise ranking method BPR~\cite{rendle2009bpr} by incorporating adversarial training; 
	\item[viii.]  \nais~\cite{NAIS}, which generalizes the well-known FISM~\cite{kabbur2013fism}  recommendation algorithm employing an attention mechanism;
\end{enumerate}

\subsubsection{Results}
Table~\ref{tab:Results} reports the top-$n$ recommendation performance of the competing approaches. The performance was measured in terms of HR@$n$, ARHR@$n$ and NDCG@$n$, focusing on the $n=10$. Model selection was performed for each dataset and metric following the procedure detailed in Section \ref{subsec:evaluationprotocol} and considering for each method the hyperparameters reported on Table~\ref{tab:Results}.

 From the deep network methods, \nais, \vae, \dae\ manage to perform very well, with the latter two showing more consistent performance across datasets. At the same time \vae, \dae\ where significantly more efficient to compute than \nais\ (50x-200x faster training time per epoch). SLIM, does remarkably well in all cases, managing to reach and surpass the performance of deep nets in most datasets (except \vae\ and \dae\ in \textit{pinterest}).  From the simple latent factor models \eigenrec\ performed well in several cases, while also being very efficient and scalable. \rpb\ did remarkably well in \movies, \books, and especially \electronics, where it was found to perform better than every method except \kstep; its performance in \movielens, \yahoo\ and \pin, however, was not competitive to the neural nets and \eigenrec. 

 We see that both variants of \recwalk\ perform very well on every metric and for all datasets. The results indicate the potential of the proposed methodology in achieving high quality top-$n$ recommendations.  
 
  Importantly, the best performing \kstep\ models also managed to ensure excellent itemspace coverage. This is in contrast to SLIM, where coverage and accuracy do not necessarily go hand-in-hand (a comparison of the coverage properties of \kstep\ and \slim\ can be found in Appendix~\ref{Appendix:Coverage}).

\subsection{Runtimes}
\label{Sec:Runtimes}

{Table~\ref{tab:Computational} reports the wall-clock timings for training \recwalk, as well as the rest of the competing methods included in Table~\ref{tab:Results}. }

\begin{table}[h!]
	\caption{Training runtimes of the competing methods on the datasets considered in Table~\ref{tab:Results}.}
	\label{tab:Computational}
	\centering
	\begin{threeparttable}
		\begin{tabular}{r cccccc}
			\toprule%
			&	{\movielens} & {\yahoo} & {\pin} & {\textit{movies\&tv}} & {\textit{books}} &  {\textit{electronics}}
			\\
			\midrule
		\myrowcolour%
		\rpb	& <1s &  <1s	& <1s  & <1s  &  <1s & <1s    \\
		\puresvd	& <1s &  <1s	& 8.1s  & <1s  &  1.5m & <1s    \\
		\myrowcolour%
		\eigenrec	& <1s &  <1s	& 6.8s  & 2.1s  &  1.5m & <1s    \\
		\vae	& 7.9m &  9.1m	& 1.7h  & 14.5m  &  2.8h & 3.9m    \\
		\myrowcolour%
		\dae	& 7.2m &  8.3m	& 1.6h  & 13.1m  &  2.7h & 3.7m    \\			
		\nais	& 6.3h &  2.1h	& 7.3h  &  4.3h  &  21.2h  & 12.6m     \\
		\myrowcolour%
		\slim	& 20.9s &  14.7s	& 14s  & 29.4s  &  2.5m & 1.1s    \\	
		\apr	& 6.9h &  4.1h	& 50h  & 5h  &  51.7h & 43.3m    \\
		\myrowcolour%
		\midrule
		\recwalk	& 2s &  1.3s	& 9s  & 2.1s  &  21s & <1s    \\
			\bottomrule
		\end{tabular}
		\begin{tablenotes}
		\footnotesize
		\item Experiments are ran on a single Intel Xeon Gold 6148 CPU @ 2.40GHz Machine with 20 cores and 64Gb DDR4 RAM. 
	\end{tablenotes}
	\end{threeparttable}
\end{table}

{ 
	Notice that \recwalk\ can be trained significantly faster than the deep-net-based alternatives---attaining runtimes comparable to highly efficient methods like \puresvd\ and \eigenrec. Further training times of \recwalk\ on datasets with large itemspaces can be found in Appendix~\ref{Appendix:LargeItemspaces}.}

\section{Related Work}
\label{Sec:Related_Work}
\textit{Neighborhood-based} methods~\cite{ning2015comprehensive} and in particular \textit{item-based} models~\cite{deshpande2004item,ning2011slim} are among the most popular and effective approaches to tackle the top-$n$ recommendation task. A notable such example is  SLIM~\cite{ning2011slim} which generalizes the traditional item-based CF approaches~\cite{sarwar2001item} to learn a sparse linear item-to-item model. SLIM has been shown repeatedly to yield state-of-the-art top-$n$ recommendation accuracy, and has inspired the proposal of several new methods in recent years~\cite{cslimzheng2014,ning2012sparse,Steck:2019:ESA:3308558.3313710}.         
Despite their success, however, item-based models are known to be negatively impacted by sparsity which can lead to decay in recommendation performance and poor itemspace coverage. 
\textit{Latent Space} methods~\cite{Cremonesi:2010:PRA:1864708.1864721,nikolakopoulos2019eigenrec,hu2008collaborative,koren2010factor,weston2013nonlinear} are particularly well-suited to alleviate such problems.  Generally speaking, the methods of this family work by projecting the elements of the recommender database into a denser subspace that captures their most salient features, giving them the ability to relate previously unrelated elements, and thus making them less vulnerable to sparsity~\cite{nikolakopoulos2019eigenrec}. 
Another family of approaches that are able to address issues related to sparsity and poor itemspace coverage are \textit{Graph-Based} methods~\cite{fouss2012experimental,kang2016top,christoffel2015blockbusters,Cooper:2014:RWR:2567948.2579244}. The innate characteristic that makes the methods of this family suited for alleviating such problems  is that they allow elements that are not directly connected, to ``influence" each other by propagating information along the edges of the underlying user-item bipartite network~\cite{NikolakopoulosWebIntelligence2015,ning2015comprehensive}. The  transitive relations captured in such a way, can be then exploited to estimate measures of proximity between the corresponding nodes~\cite{christoffel2015blockbusters,Cooper:2014:RWR:2567948.2579244} or compute similarity scores between them~\cite{fouss2012experimental}, which can be used afterwards to recommend items accordingly.  Graph-based recommendation methods relying on random walks have also been deployed in several large-scale industrial settings with considerable success~\cite{twitter2,pixie}.   \recwalk\ combines item-models with random walks, and therefore lies at the intersection of neighborhood- and graph-based methods; the inter-item transition component captures the neighborhoods of the items which are then incorporated in a random walk framework to produce recommendations.

Propelled partly by the success of deep learning in areas like speech  recognition, computer vision and natural language processing, recently, there has been a growing body of work applying neural networks to collaborative filtering; thereby extending traditional recommendation algorithms to account for non-linearities as well as to exploit neural-network related heuristics like \textit{attention mechanisms} and \textit{adversarial training}. Along these lines, a large number of methods have been proposed which extend latent space methods~\cite{ncf,vaedae,cdae}, learning-to-rank methods~\cite{APR}, factored item similarity methods~\cite{NAIS,DeepICF}. Besides targeting the pure top-n recommendation task deep nets have also been applied in sequential and session-based recommendation settings~\cite{hidasi2015session,kang2018self,tang2018personalized}, as well as in the proposal of hybrid recommendation models in which the core collaborative filtering component is augmented with available meta-information about the users or the items (for a recent survey of related techniques the reader can see~\cite{zhang2019deep}). Despite their promising performance, many neural-network based methods suffer computational and scalability issues; especially when deeper architectures are being deployed. In large scale settings this can reduce the frequency by which the underlying recommendation model can be updated and thus limit their applicability in rapidly evolving real world recommendation settings. Notable exceptions to this, include the work by Liang et al.~\cite{vaedae} that proposes  variational autoencoders to generalize linear latent-factor models showing  promising recommendation accuracy while also being computationally feasible.

The construction of \recwalk\ is inspired by the properties of nearly uncoupled Markov chains. The analysis of nearly uncoupled systems---also referred to as \textit{nearly decomposable} systems---has been pioneered by Simon~\cite{simon1961aggregation}, who reported on state aggregation in linear models of economic systems. However, the universality of Simon's ideas has permitted the theory to be used with significant success in the analysis of complex systems  arising in social sciences and  economics~\cite{ando1963essays,JOI:1358112,Sarasvathy2003203}, evolutionary biology~\cite{shpak2004simon}, cognitive  science~\cite{meunier2010modular}, administrative sciences and management~\cite{yakob2007managing,yayavaram2008decomposability}, etc. The introduction of these ideas in the fields of computer science and engineering can be traced  back to  Courtois \cite{courtois1977decomposability} who applied Simon's theory in the performance analysis of computer systems. More recently, near decomposability has been recognized as a property of the Web~\cite{blockrank} and it has inspired the development of algorithms for faster computation of PageRank~\cite{zhu2005distributed,cevahir2011site} (building on a large body of related research in the field of numerical linear algebra; see e.g.,\cite{stewart:1994introduction,stewart1991sensitivity}) as well as the development of new network centrality measures~\cite{nikolakopoulos2013ncdawarerank,AggregateRank,RSMG}. In the field of recommender systems the notion of decomposability has inspired the development of methods for incorporating meta-information about the items~\cite{nikolakopoulos2013novel,NikolakopoulosK15,NikolakopoulosWebIntelligence2015} with the blocks chosen to highlight known structural/organizational  properties of the underlying itemspace. Here, on the contrary, we exploit decomposability in the \textit{time-domain} with the blocks defined to separate the short-term from the long-term temporal dynamics of the walk in order to effect the desired mixing properties that can lead to improved recommendation performance. 

{Prior work on alleviating the devastating effects of item-popularity on the accuracy, as well as the diversity of graph-based methods, includes the hybrid method of Zhou \textit{et al.}\cite{Zhou4511}, and the re-ranking  approach proposed by Christoffel \textit{et al.}~\cite{christoffel2015blockbusters}. While the underlying goals of such methods are aligned with the motivations behind \recwalk, the approach we follow herein differs significantly. \recwalk\ leverages the spectral properties of nearly uncoupled Markov chains to enforce a time-scale dissociation of the stochastic dynamics of the walk towards equilibrium which increases the number of successive landing distributions that are still influenced by the user-specific initialization of the walk. Therefore, contrary to the aforementioned methods, our approach eliminates the need of ending the walks early---thereby allowing the walker ample time to explore the underlying network (and increase coverage), before the produced recommendation vectors start concentrating probability mass towards the popular items in the system. Furthermore, \recwalk\ introduces a novel methodology for incorporating item models in the underlying random-walk-based framework---leading to significant improvements in recommendation accuracy. 
}

\section{Conclusions and Future Directions}
\label{Sec:Conclusions}
Combining random walks with item models has the potential of exploiting better the information encoded in the item-to-item relations---leading to improved itemspace coverage and  increased top-$n$ recommendation accuracy. To gain such benefits, however, one needs to define judiciously the transition probabilities of the walks in order to counterbalance their tendency to rapidly concentrate towards the central nodes of the graph. 
To this end we introduced \recwalk; a novel random walk framework for top-$n$ recommendations that can provably provide control over convergence allowing the walk to harvest more effectively the rich network of interactions encoded within the item model on top of which it is built. Our experiments reveal that the mixing properties of the walks are indeed intertwined  with top-$n$ recommendation performance.

A very interesting direction we are currently pursuing involves the exploration of methods for statistical learning in the space of landing probabilities~\cite{AdaDIF_TSP} produced by \recwalk. Here we proposed two simple recommendation strategies to exploit these landing probabilities that were able to provide high top-$n$ recommendation accuracy, outperforming several state-of-the-art competing approaches. Our findings showcase the value of combining item-models with graph-based techniques.

\begin{acks}
This work was supported in part by NSF (1901134, 1447788, 1704074, 1757916, 1834251), Army
Research Office (W911NF1810344), Intel Corp, and the Digital Technology Center at the
University of Minnesota. Access to research and computing facilities was provided by
the Digital Technology Center and the Minnesota Supercomputing Institute.
\end{acks}

\appendix

\section{Ergodicity of the \recwalk\ chain}
\label{Appendix:Ergodicity}
\begin{proposition}[Ergodicity of \recwalk]
	\label{Prop:ergodicity}
	When $\graph{G}$ is connected and $\alpha \in (0,1)$, the Markov chain with transition probability matrix \mP\ is ergodic.
\end{proposition}
\begin{proof}
	It suffices to show that the resulting \recwalk\ chain will necessarily be irreducible and its states aperiodic and non-null recurrent~\cite{grimmett2001probability}.
	
	\begin{description}
		\item[Proof of Irreducibility:] The state space, $\set{S}$, of every Markov chain  can be partitioned uniquely as 
		\begin{displaymath}
		\set{S} = \set{T}\cup\set{C}_1\cup\set{C}_2 \cdots
		\end{displaymath}
		where $\set{C}_i$ are irreducible closed sets of recurrent states, and $\set{T}$ is the set transient states\cite{grimmett2001probability}. Therefore, to prove irreducibility it suffices to show that the transient set $\set{T}$ is empty, and that there exists a single irreducible communicating class $\set{C}$ containing all states. Let $i$ be a recurrent state and let $\set{C}$ be  the set that contains it.\footnote{Note that we can always choose such a state since Markov chains with finite state space are guaranteed to contain at least one recurrent state~\cite{grimmett2001probability}.} We will prove that starting from $i$, we can visit every other state of the \recwalk\ chain---and therefore, every state belongs to $\set{C}$. Since by assumption $\alpha \in (0,1)$ and graph $\graph{G}$ is connected, the random walk can always follow the transitions between the states contained in the stochastic matrix \mH. Notice that by definition this matrix can be seen as the transition matrix of a simple random walk on the undirected graph $\graph{G}$. Thus, when  $\alpha \in (0,1)$ the connectivity of $\graph{G}$ is enough to ensure that state $i$ can communicate with every other state in the chain. Therefore, as needed, all states belong to $\set{C}$ and the transient set $\set{T}$ is empty.     
		\item[Proof of Aperiodicity:]  The period of a state is defined as the greatest common divisor of the number of steps at which a return to said state is possible~\cite{grimmett2001probability}. Moreover, all states within the same communicating class have the same period~\cite{grimmett2001probability}. Therefore, due to the irreducibility of the \recwalk\ chain, to prove aperiodicity it suffices to show that there exists an integer $m$ and a state $i$, upon leaving which, there are positive probability random walk trajectories to return to it in $m$ steps, and also in $m+1$ steps. 
		In case of \recwalk\ this can be seen easily. Let 
		\begin{displaymath}
		w \mathdef i\to \cdots\to s \to \cdots \to i
		\end{displaymath} 
		be a positive probability trajectory from $i$ to $i$, that passes through some state $s$ that corresponds to a node $u \in \set{U}$ (notice that when $\graph{G}$ is connected and $\alpha \in(0,1)$ it is always possible to find such a  trajectory), and let $m$ be its length. We construct a new trajectory $w'$ that coincides with $w$ except for the fact that upon leaving state $s$ it follows a self-loop, thus returning to $s$, and then continues exactly like $w$. Notice that $w'$ is also a valid positive probability trajectory from $i$ to $i$ abiding by the transition probability matrix \mP\ and it has length $m+1$. 
		Therefore, it is possible to return to $i$ in $m$ and in $m+1$ steps, which means that $i$ and---consequently its communicating class $\set{C}$---is aperiodic as needed.    
	\end{description}  
	Finally, due to the finiteness of the \recwalk\ chain, we immediately get that every recurrent state is necessarily positive-recurrent~\cite{grimmett2001probability}. Therefore, the chain is ergodic and the proof is complete. 
\end{proof}

\section{Coverage}
\label{Appendix:Coverage}
Table~\ref{tab:Coverage} reports the percentage of users for who the recommendation vectors of each method support at least 50\%, and at least 90\% of the itemspace. For each method we report the coverage of the best performing model in terms of NDCG@10, from Table~\ref{tab:Results}. 
\begin{table}[!h]
	\caption{Coverage of SLIM vs \kstep.}
	\label{tab:Coverage}
	\small
	\centering
	\begin{threeparttable}
		\begin{tabular}{l r rr c rr}
			\toprule%
			\toprule%
			\centering%
			&
			& \multicolumn{2}{c}{{{SLIM}}} 
			& \phantom{abcd}
			& \multicolumn{2}{c}{{{\kstep}}}
			\\
			\cmidrule[0.4pt](lr{0.125em}){3-4}%
			\cmidrule[0.4pt](lr{0.125em}){6-7}%
			Dataset
			& 
			&  $50\%$ & $90\%$ & &  $50\%$ & $90\%$  
			\\
			\midrule
			\myrowcolour%
			\textit{movielens}\phantom{ab}  & 
			& 87.7\%            & 38.1\%    &   & 100\%  & 100\%      \\
			\textit{yahoo}      & 
			& 62.2\%            & 10.2\%   &  &  100\%  & 100\%     \\
			\myrowcolour%
			\textit{movies\&tv} & 
			& 86.5\%   & 13.9 \%   &  & 100\%   & 100\%    \\
			\textit{books}      & 
			& 25.8\%   & 0.9\%   &   & 100\%   & 100\%     \\
			\myrowcolour%
			\myrowcolour%
			\textit{electronics} & 
			& 73.2\%                                       & 3.8\%   &    & 100\%            & 100\%       \\
			\textit{pinterest}      & 
			& 0.02\%                                        & 0.0\%    &  & 100\%            & 100\%      \\
			\bottomrule
			\bottomrule
		\end{tabular}
	\end{threeparttable}
\end{table}

Note that while in case of \pr\ coverage is guaranteed for any $\eta \in (0,1)$, for \kstep\ coverage is controlled by the selected number of steps $K$. Our results indicate that the $K$ that leads to better recommendation accuracy also does an excellent job covering the itemspace in all datasets. In case of SLIM, on the contrary, we see that coverage and recommendation accuracy do not necessarily accord.

\section{Additional Runtimes for Larger Datasets}
\label{Appendix:LargeItemspaces}
To further assess the computational efficiency of training \recwalk\ in settings where the number of items in the system is large, we train \recwalk\ on seven larger snapshots of the Amazon dataset~\cite{He:2016:UDM:2872427.2883037} and we report the resulting runtimes in Table~\ref{tab:runtimes2}.
\begin{table}[h]
	\caption{ Training times of \recwalk\  for larger datasets. }
	\centering
	\small
	\begin{threeparttable}
		\begin{tabular} {lrrrr   }
			\toprule
			\toprule
			Name                & \#users & \#items & density & {Training times} \\
			\midrule
			\myrowcolour%
			\textit{Grocery\_and\_Gourmet\_Food}  & 768,438   & 166,049   & $1.02\times10^{-5}$  
			& 2.3m  \\	
			\textit{Beauty}  & 1,210,271   & 249,274   & $6.71\times10^{-6}$  
			& 5.9m  \\	
			\myrowcolour%
			\textit{Health\_and\_Personal\_Care}      & 1,851,132   & 252,331   & $6.38\times10^{-6}$    & 8.9m \\
			\textit{Digital\_Music} & 478,235  & 266,414   & $6.56\times10^{-6}$  
			& 2.5m  \\		
			\myrowcolour%
			\textit{Cell\_Phones\_and\_Accessories}  & 2,261,045   & 319,678   & $4.77\times10^{-6}$  
			& 13.3m  \\		
			\textit{CDs\_and\_Vinyl} & 1,578,597  & 486,360   & $4.88\times10^{-6}$  
			& 20.1m  \\			
			\myrowcolour%
			\textit{Clothing\_Shoes\_and\_Jewelry}      & 3,117,268   & 1,136,004   & $1.62\times10^{-6}$    & 98.9m \\	

			\bottomrule
			\bottomrule 
		\end{tabular}
		\begin{tablenotes}
			\footnotesize
			\item Experiments are ran on a single Intel Xeon Gold 6148 CPU @ 2.40GHz Machine with 20 cores and 64Gb DDR4 RAM. Parameters $\gamma_1,\gamma_2$ are fixed to the value 10. $C$ is fixed to 200.
		\end{tablenotes}
	\end{threeparttable}
	\label{tab:runtimes2}
\end{table}

Our results indicate that \recwalk\ can be trained efficiently even for large itemspaces. 

\section{Performance for Different Levels of Sparsity}
\label{Appendix:Sparsity}
{To assess the quality of \recwalk\ in the presence of sparsity, we conduct the following experiment: we take the  \movielens\ data and we create four artificially sparsified versions of the dataset by randomly selecting to include 40--85\% of the available data. We train RecWalk in each setting and in Fig.\ref{fig:sparsityperformance} we report its out-of-sample performance in terms of NDCG@10.  For reference we also report \recwalk's performance when all available data are exploited (``100\%'' bar in Fig.\ref{fig:sparsityperformance}). } 
\begin{figure}[h!]
	\centering
	\includegraphics[width=0.7\linewidth]{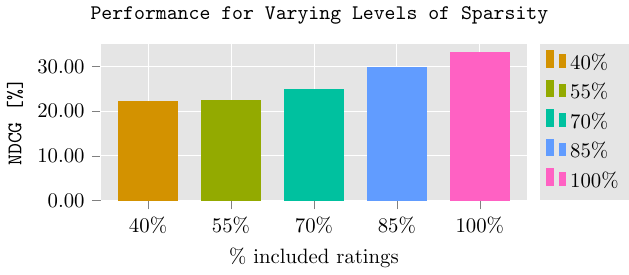}
	\caption{The figure reports the performance of \recwalk\ on the \movielens\ data, under different sparsity settings. Model selection was performed for each setting following the procedure detailed in Section \ref{subsec:evaluationprotocol} and by considering the hyperparameters: (fixed) $C = 200$,  (fixed) $\alpha = 0.005$, $\gamma_1 \in \{1,3,5,10\},\gamma_2 \in \{0.1, 0.5, 1, 3, 5, 7, 9, 11, 15, 20\}$ and $K \in \{3, \dots, 20\}$.}
	\label{fig:sparsityperformance}
\end{figure}

{We find that \recwalk\ manages to retain more that 2/3 of its full recommendation accuracy even in the case where only 40\% of the available feedback is exploited (specifically, \recwalk's NDCG@10 on the sparsest case is 33.2\% lower  than its corresponding performance when all available data are used).  
	Expectedly, as the percentage of included interactions grows, the quality of the trained model improves and the recommendation accuracy increases significantly. 
	Finally, we find that full itemspace coverage is attained for every user in the system even in the sparsest setting. 
}

\end{document}